\documentclass{article}

\usepackage[T1]{fontenc}
\usepackage[utf8]{inputenc}
\usepackage{microtype}

\usepackage{graphicx} 
\usepackage{subcaption} 
\usepackage{amsmath,amsthm,amssymb} 
\usepackage{hyperref}
\usepackage[nameinlink,capitalize]{cleveref}
\usepackage{authblk}
\usepackage{enumitem}

\newtheorem{theorem}{Theorem}

\title{Flow-Based Modelling of Population Dynamics with Consecutive Continuous Mutations}
\author[1,2,3]{Alexander Bratus}
\author[4]{Tatiana Yakushkina}
\author[3]{Vladimir Posvyanski}
\affil[1]{Marchuk Institute of Numerical Mathematics, Russian Academy of Sciences, 119333 Moscow, Russia}
\affil[2]{Moscow Center for Fundamental and Applied Mathematics at INM RAS, 119333 Moscow, Russia}
\affil[3]{Institute of Management and Digital Technologies, Russian University of Transport, 127055 Moscow, Russia}
\affil[4]{A.I. Alikhanyan National Science Laboratory (Yerevan Physics Institute) Foundation, 375036 Yerevan, Armenia}

\date{August 2025}

\begin{document}

\maketitle

\begin{abstract}
This study presents a continuous mathematical model of population dynamics incorporating the sequential emergence of new genotypes under limited resources. The model describes the evolution of genotype density in mutation space and time through a nonlinear flow-type equation that couples transport driven by a time-dependent mutation rate with logistic growth and nonlocal competition. Analytical solutions are derived for the advection--reaction regime without reverse mutations by applying the method of characteristics, yielding explicit expressions for variable carrying capacities and mutation velocities. We analyse the influence of decaying and accelerating mutation rates on the saturation and propagation of population fronts using level-set geometry. We demonstrate that the inclusion of reverse mutations transforms the system into a quasilinear parabolic equation with diffusion in genotype space, whose numerical solutions reveal the stabilising and smoothing effects of backward mutation flows. The proposed framework generalises the classical quasispecies and Crow--Kimura models by introducing logistic regulation, variable mutation rates and reversible transitions, thereby providing a unified mathematical approach to the analysis of evolutionary dynamics in virology, bacterial adaptation and tumour progression.
\end{abstract}


\section{Introduction}

The population dynamics of many biological species is characterised by the occurrence of successive mutations, which give rise to novel genotypes. Examples include the emergence of new virus genotypes in HIV infection and during the COVID-19 pandemic~\cite{Martinez2011, Blyuss2024}. Similar phenomena are observed in the dynamics of cancer cell populations~\cite{Loeb2000, Loeb2011} and, more generally, rapidly mutating pathogens~\cite{Bessonov2020}. For instance, sequential mutations in pancreatic cancer cells can lead to the emergence of multiple distinct genotypes~\cite{Bratus2022}. The mathematical description of such processes is challenging because it requires the simultaneous treatment of two fundamentally different dynamics: the discrete process of the probabilistic appearance of new genotypes and the continuous process of population growth. 

In quasispecies models, such as those proposed by Eigen and by Crow and Kimura, discrete mutation and reproduction events are implemented via mutation matrices, which define the probabilities of generating new genotypes. However, these models do not account for species mortality or competitive interactions and typically assume that the total population size of all species remains constant over time. 

One way to overcome these limitations is to transition to continuous models. In this work, we propose a model of the sequential emergence of new genotypes based on a flow equation~\cite{Zeldovich1972, Bratus2010}. In this formulation, the population dynamics is described under the following assumptions:
\begin{enumerate}
    \item The set of genotypes is conceptually unbounded, each genotype assigned an index of the set $\{0, 1, \dots\}$.
    \item Successive mutations and the emergence of new genotypes occur at a prescribed positive mutation rate that depends on time.
    \item The growth rates of the genotypes depend both on the index of the genotype and on time.
    \item The population size of each genotype follows a logistic-type regulation.
    \item The limiting density of genotypes and their total population size are bounded by a function that depends on both the genotype index and time.
    \item In some cases, sequential mutations may admit reverse mutations.

Under the stated assumptions, it can be shown that the dynamics of a population with the sequential emergence of new genotypes can be described by either an initial boundary value problem for a flow equation with reverse mutations or by a quasilinear parabolic equation with nonlocal components. The model can be generalised to incorporate multiple branching sequences of successive mutations.
\end{enumerate}

\section{Model Formulation}

\subsection{Problem statement}
Let \( N(x,t) \) denote the number of genotypes with index \( x \in [0,\infty) \) at time \( t \geq 0 \).
For an interval of consecutive indices \((x, x+\Delta x]\) with \(\Delta x > 0\), define \emph{flow density}
\[
u(x,t) := \lim_{\Delta x \to 0} \frac{N(x+\Delta x, t) - N(x,t)}{\Delta x}.
\]
Let \( \Phi(x,t) \) denote the change in population rate at index \(x\) per unit of time,
\[
\Phi(x,t) := \lim_{\Delta t \to 0} \frac{N(x, t+\Delta t) - N(x,t)}{\Delta t}.
\]
If, within \((x, x+\Delta x]\) over \((t, t+\Delta t]\), no creation or annihilation occurs, then
\begin{equation}
N(x,t) = - \int_{x}^{x+\Delta x} u(\xi,t) \, d\xi + \int_{x}^{x+\Delta x} u(\xi, t+\Delta t) \, d\xi, \quad \Delta x > 0.
\label{eq:integrand1}
\end{equation}
Equivalently, in terms of $\Phi$,
\begin{equation}
N(x,t) = - \int_{t}^{t+\Delta t} \Phi(x+\Delta x, \tau) \, d\tau + \int_{t}^{t+\Delta t} \Phi(x, \tau) \, d\tau, \quad \Delta t > 0.
\label{eq:integrand2}
\end{equation}

The integrands in \eqref{eq:integrand1}, \eqref{eq:integrand2} can be written as  
\[
\frac{\partial u(x,\bar t)}{\partial {t}} \, \Delta x \, \Delta t
\quad \text{and} \quad
-\frac{\partial \Phi(\bar{x},t)}{\partial x} \, \Delta x \, \Delta t,
\]
with $\bar{x} \in (x, x+\Delta x]$, $\bar{t} \in (t, t+\Delta t]$.  
Letting $\Delta x,\Delta t \to 0$ yields the flow continuity equation~\cite{Zeldovich1972}:
\begin{equation}
\frac{\partial u(x,t)}{\partial t} + \frac{\partial \Phi(x,t)}{\partial x} = 0.
\label{eq:flow_eq}
\end{equation}

If there is a birth/death source in \((x, x+\Delta x]\) during \((t, t+\Delta t]\) that depends on the population density via a function $F(u(x,t))$, the balance includes
\[
\int_{t}^{t+\Delta t} \int_{x}^{x+\Delta x} F(u(\xi,\tau))\, d\xi\, d\tau,
\]
and \eqref{eq:flow_eq} becomes
\begin{equation}
\frac{\partial u(x,t)}{\partial t} + \frac{\partial \Phi(x,t)}{\partial x} = F(u(x,t)).
\label{eq:flow_with_F}
\end{equation}

Let $v(t)$ be the rate of successive mutations. If
\begin{equation*}
\Phi(x,t) = v(t)\, u(x,t),
\label{eq:phi_adv}
\end{equation*}
then \eqref{eq:flow_with_F} reduces to a first-order advection equation:
\begin{equation}
\frac{\partial u(x,t)}{\partial t} + v(t) \frac{\partial u(x,t)}{\partial x} = F(u(x,t)).
\label{eq:advdiff}
\end{equation}

If reverse mutations (from higher to lower density) are present~\cite{Abolitsina1960}, we set
\begin{equation*}
\Phi(x,t) = v(t)\, u(x,t) - \mu^2 \frac{\partial u(x,t)}{\partial x},
\end{equation*}
with $\mu>0$ characterising the intensity of reverse mutations. In that case,
\begin{eqnarray}
\frac{\partial u(x,t)}{\partial t}& + v(t) \displaystyle\frac{\partial u(x,t)}{\partial x} =
 F(u(x,t)) + \mu^2 \frac{\partial^2 u(x,t)}{\partial x^2},\label{eq:advdiff_mu}
\\  &x \in [0,\infty), \quad t \ge 0, \nonumber
\end{eqnarray}
which includes advection and diffusion.

To solve \eqref{eq:advdiff}, we specify initial and boundary data:
\begin{equation}
u(x,0) = \varphi(x), \qquad u(0,t) = g(t), \qquad t \ge 0.
\label{eq:ibc}
\end{equation}

\subsection{Two choices for the population change source term}

\paragraph{Local logistic regulation.}
\begin{equation}
F(u) = r(x,t) \big( q(x) - u(x,t) \big) u(x,t),
\label{eq:F_logistic}
\end{equation}
where $r(x,t)$ is the growth rate of genotype $x$ at time $t$, and $q(x)$ is a limiting density. The function $q(x)$ enables modelling the effect of intervention, e.g., preventive, sanitary, and therapeutic measures, on population dynamics.

\paragraph{Nonlocal competition.}
\begin{equation}
F(u) = r(x,t) \left( q(x) - \gamma \int_0^x u(\xi,t)\,d\xi - u(x,t) \right) u(x,t),
\label{eq:F_nonlocal}
\end{equation}
where the nonlocal term accounts for competitive pressure accumulated along the mutation axis.

\paragraph{Multiple branches.}

If the root genotype $x=0$ spawns several mutation branches with growth rates $r_i(x_i,t)$ and mutation rates $v_i(t)$, $i=1,\dots,n$, the corresponding equations mirror \eqref{eq:advdiff_mu}–\eqref{eq:ibc} for each branch.

\section{Analytical solution without reverse mutations}

Consider \eqref{eq:advdiff_mu},\eqref{eq:ibc} with $\mu=0$ and $v(t)$ --- a continuous function of $t$. Define
\[
V(t) = \int_0^t v(\tau)\, d\tau.
\]

\begin{theorem}\label{th:t1}
Let the growth rate function $r(x,t)\ge 0$ be continuous on $[0,\infty)\times[0,\infty)$, boundary data $\varphi(x)$ and $g(t)$ defined by smooth functions with initial values $\varphi(0)=g(0)$, $\varphi'(0)=g'(0)$, and parameter $q\in\mathbb{R}$ be constant. Assume $V$ is monotone on $t\ge 0$. Then:
\begin{itemize}
\item[\textbf{(i)}]
\noindent In $D_1=\big\{(x,t): x\ge V(t),\ t\ge0\big\}$,
\begin{equation}
u_1(x,t) = \frac{q \, \varphi(y_1)}{\varphi(y_1) + \left(q - \varphi(y_1)\right) e^{-Z(x,t)}},
\qquad y_1 = x - V(t),
\label{eq:u1}
\end{equation}
where $Z$ solves
\begin{equation}
\frac{\partial Z(x,t)}{\partial t} + v(t) \frac{\partial Z(x,t)}{\partial x} = r(x,t) q, 
\qquad Z(x,0) = 0.
\label{eq:Z}
\end{equation}

\item[\textbf{(ii)}]
In $D_2=\big\{(x,t): x< V(t),\ t\ge0\big\}$,
\begin{flalign}
    u_2(x,t) =& \frac{q \, g(y_2)}{g(y_2) + \left(q - g(y_2)\right) 
e^{-Z(x,t) + Z(0,\, V^{-1}(\tau-x))}}, \label{eq:u2}\\
 &y_2 = V^{-1}(\tau - x),\ \tau=V(t).\nonumber
\end{flalign}
\end{itemize}
\end{theorem}

\begin{proof}

Part\textbf{(i)}  is verified by direct substitution of \eqref{eq:u1} into \eqref{eq:advdiff_mu} with \eqref{eq:ibc}.

For deriving \textbf{(ii)}, let 
$$\vartheta=\tau-x \text{ and } \gamma=V^{-1}(\vartheta).$$
For any smooth function $f$,
\[
\frac{\partial}{\partial t} f\!\big(V^{-1}(\vartheta)\big)
= f'_{\gamma}\, V^{-1}_{\vartheta}(\vartheta)\, v(t),
\qquad
\frac{\partial}{\partial x} f\!\big(V^{-1}(\vartheta)\big)
= - f'_{\gamma}\, V^{-1}_{\vartheta}(\vartheta).
\]
Lower indices indicate partial derivatives in the corresponding variables.  
Let
\[
P\!\big(V^{-1}(\vartheta)\big) = Z\!\big(0, V^{-1}(\vartheta)\big).
\]
Then 
\[\frac{\partial}{\partial t} P + v(t)\,\frac{\partial}{\partial x} P=0.\] 
along characteristics.  
Define
\[
\mu(x,t) = Z(x,t) - Z\!\big(0, V^{-1}(V(t)-x)\big).
\]
Using \eqref{eq:Z} one checks that in $D_2$, $\mu$ solves
\[
\frac{\partial \mu}{\partial t} + v(t)\frac{\partial \mu}{\partial x} = r(x,t)\,q,
\qquad \mu(x,0)=0,
\]
and substituting back gives \eqref{eq:u2}.
\end{proof}

\subsection{Explicit $Z$ and $\mu$ for selected $r(x)$}

\paragraph{Monotone growth:}
\[
r(x) = \frac{\delta}{1+x}, \quad \delta>0,\ v(t)\equiv v_0>0.
\]
Then, for $(x,t)\in D_1$,
\[
Z(x,t) = \delta t - \frac{\delta}{v_0}\ln\!\left(\frac{1+x}{1+x-v_0 t}\right),
\qquad
\mu(x,t) = \frac{\delta x}{v_0} - \frac{\delta}{v_0}\ln(1+x) \quad \text{in } D_2.
\]

\paragraph{Non-monotone growth:}
\[
r(x) = \frac{\delta_1 x}{1+\delta_2 x^2}, \quad \delta_1,\delta_2>0,\ v(t)\equiv v_0>0.
\]
Then, for $(x,t)\in D_1$,
\begin{flalign*}
&Z(x,t) = \frac{\delta_1}{2 v_0 \delta_2} 
\ln\!\left( \frac{1+\delta_2 (x-v_0 t)^2}{1+\delta_2 x^2} \right),
\\
&\mu(x,t) = -\frac{\delta_1}{2 v_0 \delta_2} \ln\!\big(1+\delta_2 x^2\big) \ \text{in } D_2.
\end{flalign*}

\subsection{Variable mutation rate cases}

Consider two extreme scenarios,
\[
v(t) = v_0 e^{-\alpha t} \quad \text{and} \quad v(t) = v_0 e^{\alpha t}, \qquad \alpha>0,
\]
with \( r(x) = \dfrac{\delta x}{1+x},\ \delta>0\).

\subsubsection{Case $v(t)=v_0 e^{-\alpha t}$}
We get

\begin{flalign}
&x > \dfrac{v_0}{\alpha}\big(1-e^{-\alpha t}\big), \nonumber\\
&Z(x,t) = \delta \left[
t - \frac{1}{1+x + \frac{v_0}{\alpha} e^{-\alpha t}}
\left(
t + \frac{1}{\alpha} \ln(1+x) \right.\right.\\
&\left. \left.- \frac{1}{\alpha} 
\ln\!\left(1+x+\frac{v_0}{\alpha} (e^{-\alpha t} - 1)\right)
\right)
\right],
\end{flalign}
and

\begin{flalign}
&x < \dfrac{v_0}{\alpha}\big(1-e^{-\alpha t}\big),\nonumber\\
&\mu(x,t) = \delta \left[
t + \frac{1}{\alpha} \ln\!\left( e^{-\alpha t} + \frac{\alpha}{v_0} x \right)
\right] \\
&\quad - \frac{\delta}{1+x+\frac{v_0}{\alpha} e^{-\alpha t}}
\left[t + \frac{1}{\alpha} \ln\!\left( e^{-\alpha t} + \frac{\alpha}{v_0} x \right) 
+ \frac{1}{\alpha} \ln(1+x) \right].\nonumber
\end{flalign}

\subsubsection{Case $v(t)=v_0 e^{\alpha t}$}
Obtain $Z,\mu$ from the previous case by replacing $\alpha\mapsto -\alpha$.

\section{Level curves for different sets of parameters}

Figure~\ref{fig:p1a} shows level curves of the solution to \eqref{eq:advdiff_mu}–\eqref{eq:ibc} with
\[
\varphi(x) = \varepsilon e^{-x^2/\varepsilon^2},\quad 
\quad
\varepsilon = 0.5,\ \delta = 0.15,\ v_0=0.3,\ q=1,
\]
where the boundary influx decays monotonically towards $\varepsilon_1$ over time
\[ g(t) = \frac{\varepsilon + \varepsilon_1 t}{1+t}, \varepsilon_1 = 0.3\]
as well as monotonic function \( r(x) = \dfrac{\delta x}{1+x} \).
Figure~\ref{fig:p1b} repeats the setting with \(\varepsilon_1 = 0 \).
\begin{figure}[!h]
  \centering
  \begin{subfigure}[t]{0.45\linewidth}
    \centering
    \includegraphics[width=\linewidth]{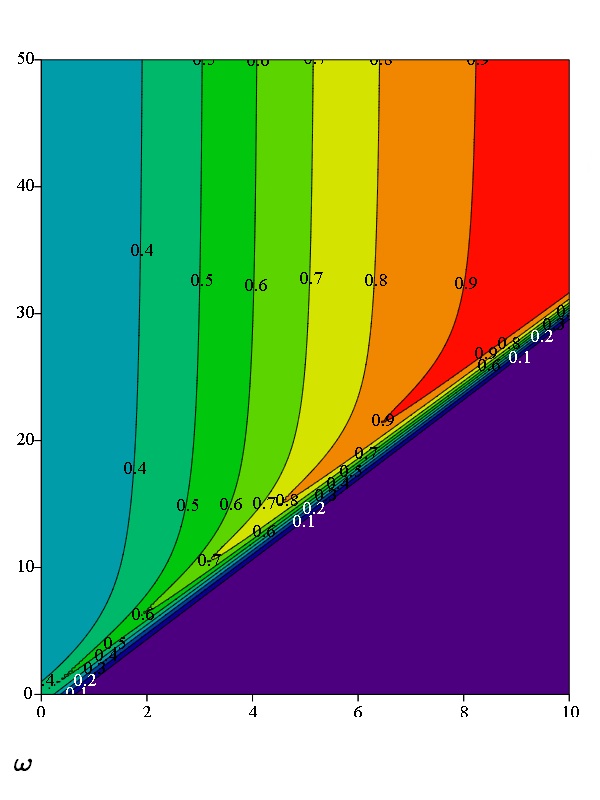}
    \caption{Case: the boundary influx $g(t)$ starts relatively high and increases with time.}
    \label{fig:p1a}
  \end{subfigure}
  \hfill
  \begin{subfigure}[t]{0.45\linewidth}
    \centering
    \includegraphics[width=\linewidth]{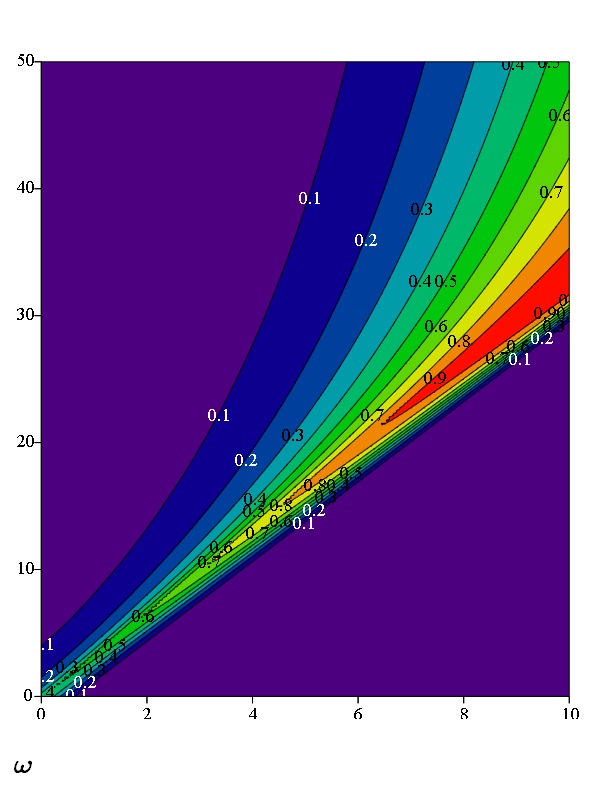}
    \caption{Case: boundary influx decays over time.}
    \label{fig:p1b}
  \end{subfigure}
  \caption{Level curves of the solution $u(x,t)$ to \eqref{eq:advdiff_mu}-\eqref{eq:ibc} under two different forms of the function $g(t)$. Here and after, the horizontal axis corresponds to the variable $x$, the vertical axis to time $t$, and the colour shading indicates the values of $u(x,t)$. Contour lines connect points where the solution attains the same value $\omega$, visualising the evolution of genotype density over time. }
  \label{fig:p1}
\end{figure}

\begin{figure}[!h]
  \centering
  \begin{subfigure}[t]{0.45\linewidth}
    \centering
    \includegraphics[width=\linewidth]{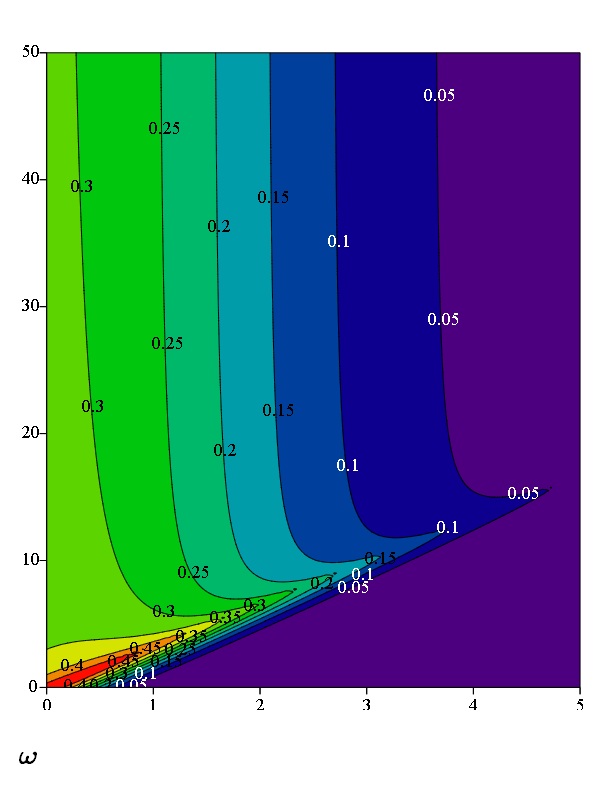}
    \caption{Case: the boundary influx $g(t)$ starts relatively high and increases with time; with $p=0.3$, solutions sustain broader growth across $\omega, T$.}
    \label{fig:p2a}
  \end{subfigure}
  \hfill
  \begin{subfigure}[t]{0.45\linewidth}
    \centering
    \includegraphics[width=\linewidth]{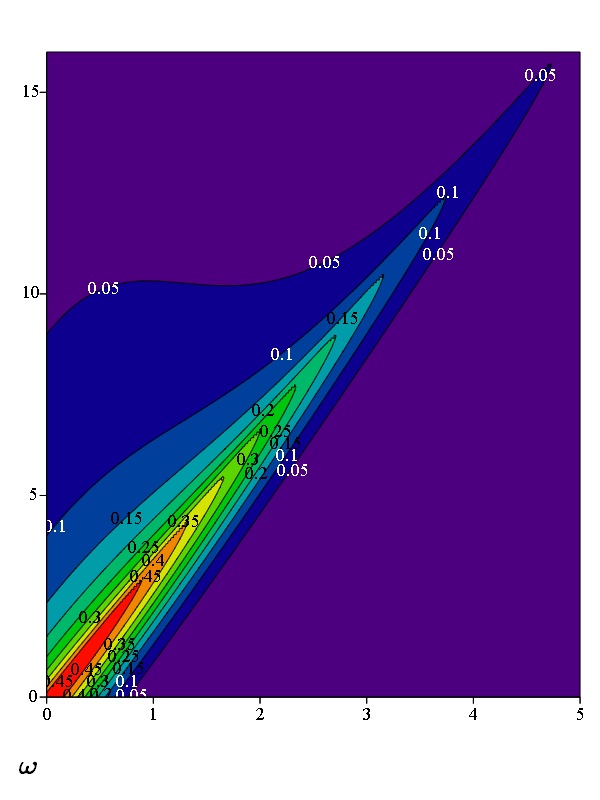}
    \caption{Case: boundary influx decays over time with $p=0$, persistence is weaker and restricted.}
    \label{fig:p2b}
  \end{subfigure}
  \caption{Level curves of the solution $u(x,t) = \omega$ to \eqref{eq:advdiff_mu}-\eqref{eq:ibc} under two different forms of the function $g(t)$. }
  \label{fig:p2}
\end{figure}

Figures~\ref{fig:p2} (a,b)  illustrate the same problem when \( r(x) = \dfrac{\delta_1 x}{1+\delta_2 x^2} \)  where \(\delta_1 = 0.15, \delta_2 = 0.1 \) and two different values $\varepsilon_1=0.3$ and $\varepsilon_1=0$.

The results suggest that the growth of the founder genotype ($x=0$) is the primary driver of the overall population dynamics, whereas the initial distribution $\varphi(x)$ plays a lesser role. Decreasing the mutation rate increases the total population size, while increasing it reduces genotype counts. Specifically, it allows for more accumulation in existing genotypes, resulting in a larger total population at any given time. Whereas increasing $v(t)$ rapidly moves individuals to a higher index (diluting the population density across many genotypes), which is why high $v(t)$ reduces counts per genotype. 

Figures~\ref{fig:p3}-\ref{fig:p4} illustrate the cases with variable mutation rate, discussed in section \textbf{3.2}: decreasing $v(t)$ leads to horizontal bending of level curves, reflecting saturation of transport with time; increasing $v(t)$ drives acceleration, yielding a steep diagonal propagation. It illustrates well how transport dynamics can dominate the shape of the solution.

\begin{figure}[!h]
  \centering
  \begin{subfigure}[t]{0.45\linewidth}
    \centering
    \includegraphics[width=\linewidth]{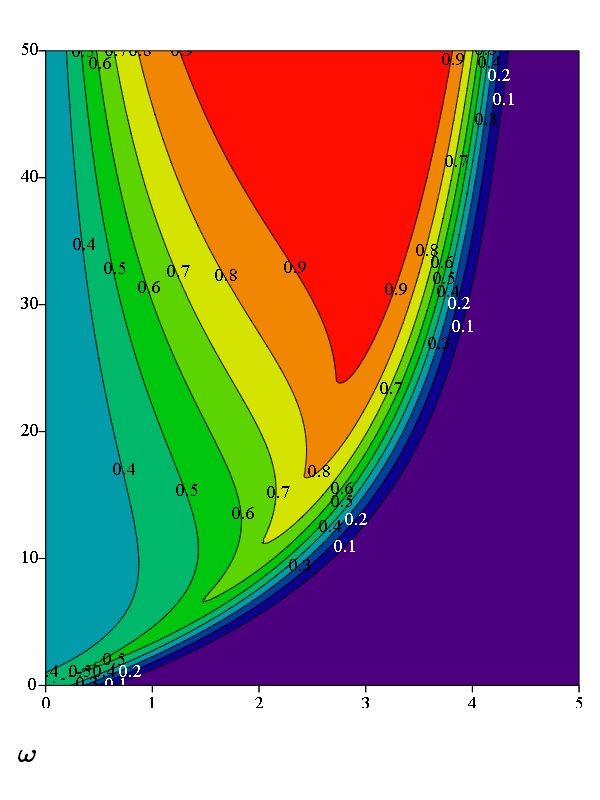}
    \caption{Case:  $v(t)=v_0 e^{-\alpha t}$ (decaying mutation flow).}
    \label{fig:p3a}
  \end{subfigure}
  \hfill
  \begin{subfigure}[t]{0.45\linewidth}
    \centering
    \includegraphics[width=\linewidth]{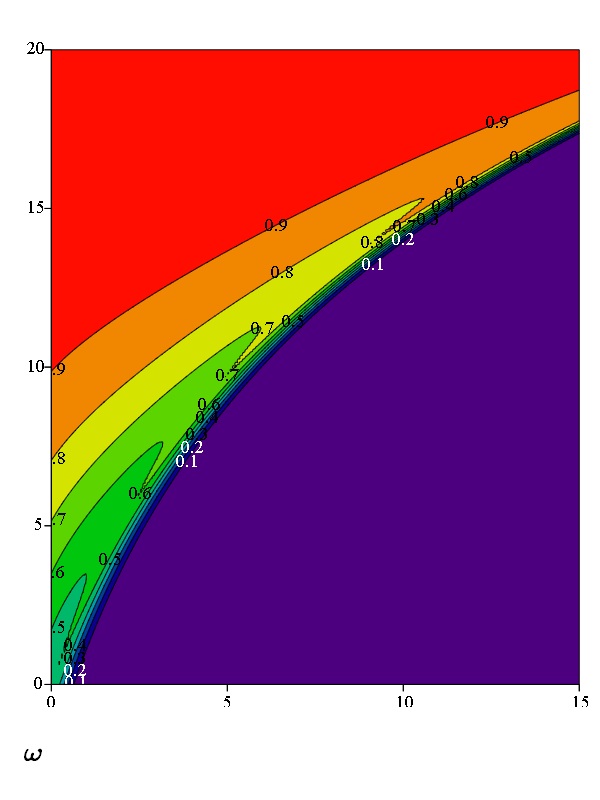}
    \caption{Case: $v(t)=v_0 e^{+\alpha t}$ (growing mutation flow).}
    \label{fig:p3b}
  \end{subfigure}
  \caption{Level sets of the solution \(u(x,t)\):
  (a) Decaying \(v(t)\) produces horizontal bending of level curves,
  reflecting the saturation of transport with time.
  (b) Growing \(v(t)\) drives acceleration, yielding a steep diagonal propagation.}
  \label{fig:p3}
\end{figure}
\begin{figure}[!h]
  \centering
  \begin{subfigure}[t]{0.45\linewidth}
    \centering
    \includegraphics[width=\linewidth]{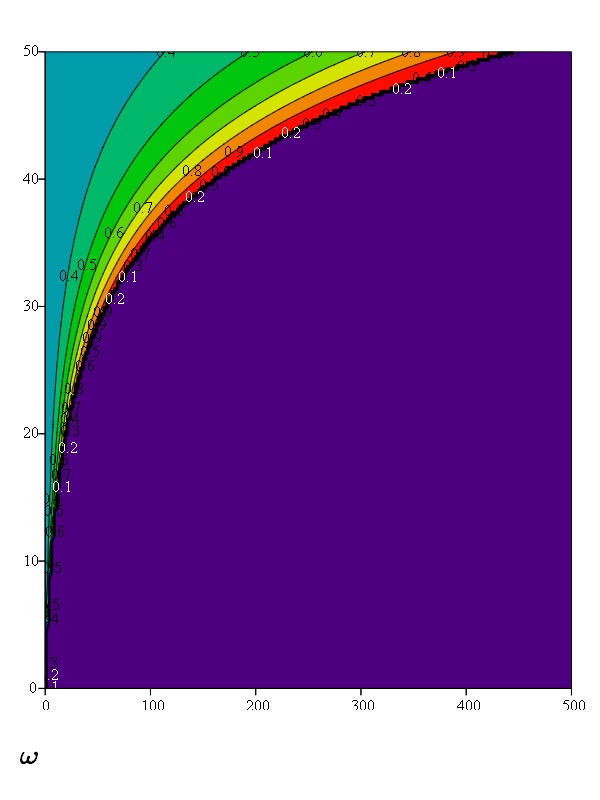}
    \caption{Case: the boundary influx $g(t)$ has $\varepsilon_1=0.3$, strong accelerated expansion.}
    \label{fig:p3c}
  \end{subfigure}
  \hfill
  \begin{subfigure}[t]{0.45\linewidth}
    \centering
    \includegraphics[width=\linewidth]{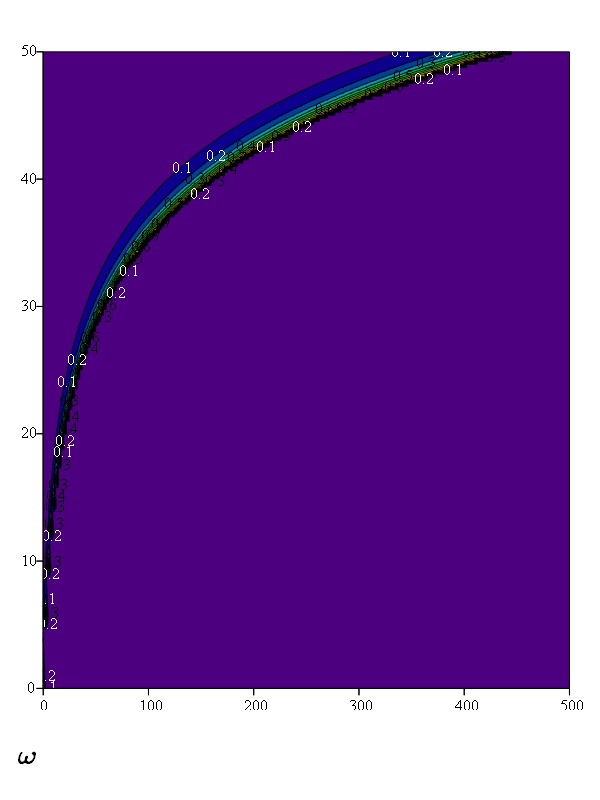}
    \caption{Case: boundary influx with $\varepsilon_1=0$, weak accelerated expansion.}
    \label{fig:p3d}
  \end{subfigure}
  \caption{Level sets of the solution $u(x,t)$ with exponentially growing transport velocity 
  $v(t) = v_{0} e^{+\alpha t}$. Boundary input at $x=0$ with $\varepsilon_1=0$ produces a thinner wavefront with fewer high-$\omega$ levels propagating in $x$. }
  \label{fig:p4}
\end{figure}

\section{Non-constant limiting density $q(x)$}

\begin{theorem}\label{th:qvar}
Assume the conditions of \cref{th:t1} and additionally $\varphi(0)=g(0)=0$. Let $q(x)$ be smooth on $[0,\infty)$. Then:

\begin{itemize}
    \item [\textbf{(i)}] In $D_1=\big\{x\ge V(t)\big\}$,
\begin{equation}
u_1(x,t) =
\frac{q(y_1)\, \varphi(y_1)}
{\varphi(y_1) + \left( q(y_1) - \varphi(y_1) \right) e^{-Z_1(x,t)}},
\qquad y_1 = x - V(t),
\label{eq:u1_q}
\end{equation}
where $Z_1$ solves
\begin{flalign}
    \frac{\partial Z_1}{\partial t}& + v(t) \frac{\partial Z_1}{\partial x} =
q(x)\, r(x,t) +
\frac{r(x,t)\, \varphi(y_1)\, \big( q(x) - q(y_1) \big)}
{q(y_1) - \varphi(y_1)}
\, e^{Z_1(x,t)},\label{eq:Z1}\\
 &Z_1(x,0)=0.\nonumber
\end{flalign}

\item [\textbf{(ii)}]  In $D_2=\big\{x< V(t)\big\}$,
\begin{equation}
u_2(x,t) =
\frac{q(y_2)\, g(y_2)}
{g(y_2) + \left( q(y_2) - g(y_2) \right) e^{-Z_2(x,t)}},
\qquad y_2 = V^{-1}(V(t)-x),
\label{eq:u2_q}
\end{equation}
where $Z_2$ solves
\begin{flalign}
\frac{\partial Z_2}{\partial t} &+ v(t) \frac{\partial Z_2}{\partial x} =
r(x,t)\, q(x) +
\frac{r(x,t)\, g(y_2)\, \big( q(x) - g(y_2) \big)}
{q(y_2) - g(y_2)} \, e^{Z_2(x,t)},\label{eq:Z2}\\
& Z_2(0,t)=0.\nonumber
\end{flalign}
\end{itemize}
\end{theorem}

\begin{proof}
Substitute the ansatz \eqref{eq:u1_q} in $D_1$ and \eqref{eq:u2_q} in $D_2$ into \eqref{eq:advdiff_mu} and compare coefficients along characteristics.  Equations \eqref{eq:u1_q} and \eqref{eq:Z1} can be reduced to Bernoulli-type equations \eqref{eq:u2_q} and \eqref{eq:Z2},
which admit analytical solutions. This reduction is important, since it allows
explicit evaluation of the function \(q(t)\), which characterizes the influence
of preventive and sanitary measures on the population dynamics.

Let $w_1=\exp(Z_1)$. Along $x-V(t)=y_1$,
\begin{flalign}
\frac{d w_1}{dt} =&
r( y_1 + V(t), t)\, q( y_1 + V(t))\, w_1
+ \\\nonumber
&r( y_1 + V(t), t)\, \varphi(y_1)\,
\frac{q(y_1 + V(t)) - \varphi(y_1)}{q(y_1) - \varphi(y_1)} \, w_1^2,
\\&w_1(y_1,0)=1.\nonumber
\end{flalign}
Similarly for $w_2=\exp(Z_2)$ along $t - V^{-1}(x)=y_2$ one obtains \eqref{eq:Z2}.
\end{proof}

Figure~\ref{fig:p4a}  shows the level surface for \eqref{eq:advdiff_mu}–\eqref{eq:ibc} with
\[
\varphi(x) = x e^{-(x/\varepsilon)^2},\quad g(t) = \frac{\varepsilon_1 t}{1+t},\quad
r(x) = \frac{\delta x}{1+\delta_1 x^2},
\]
parameters
\[
\varepsilon=0.5,\quad \varepsilon_1=0.3,\quad \delta=0.15,\quad \delta_1=0.1,\quad v_0=0.3,
\]
and limiting density
\begin{equation}
q(x) = \frac{1 + 0.4\,x}{1 + x}.
\label{eq:qexample}
\end{equation}
Figure~\ref{fig:p4b}  shows the same with
\[
g(t) = 0.15\,(1-\cos t).
\]
These plots illustrate that constraints on $q(x)$ substantially affect total genotype number and allow one to encode preventive, sanitary, or therapeutic impacts in a parsimonious way.
\begin{figure}[!h]
  \centering
  \begin{subfigure}[t]{0.45\linewidth}
    \centering
    \includegraphics[width=\linewidth]{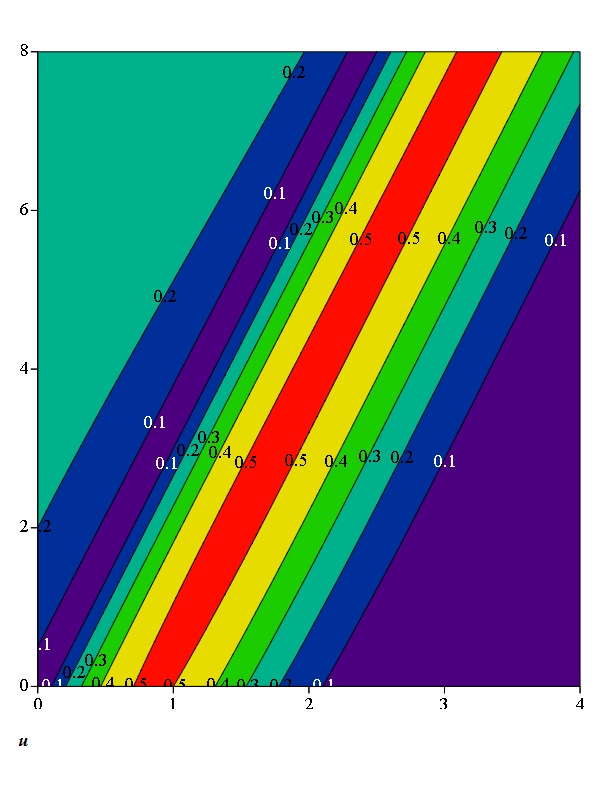}
    \caption{Case:  boundary condition \\ $g(t)=\tfrac{\varepsilon_1 t}{1+t}$.}
    \label{fig:p4a}
  \end{subfigure}
  \hfill
  \begin{subfigure}[t]{0.45\linewidth}
    \centering
    \includegraphics[width=\linewidth]{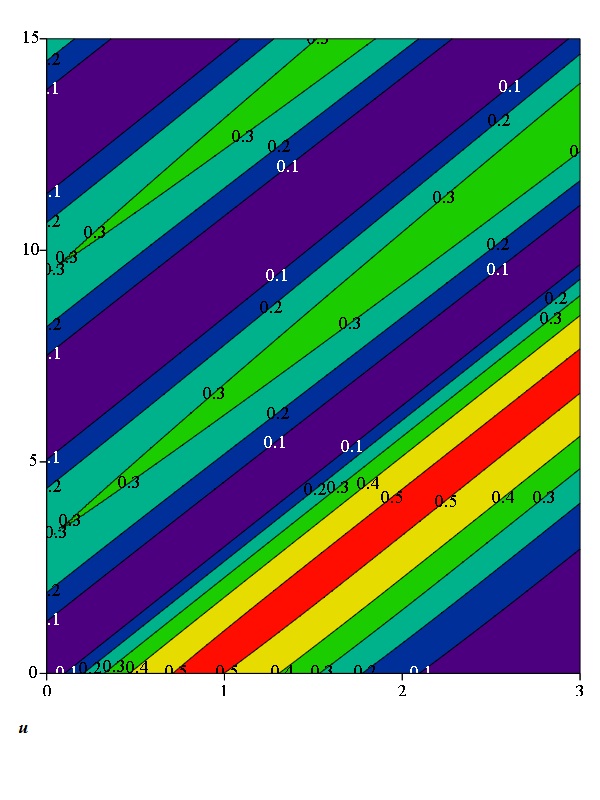}
    \caption{Case: oscillatory inflow \\$g(t)=0.15(1-\cos t)$.}
    \label{fig:p4b}
  \end{subfigure}
  \caption{Level sets of the solution $w(x,t)$: 
  (a) With monotone inflow, propagation is steady and evenly spaced.  
  (b) With oscillatory inflow, the level sets display periodic modulation, indicating alternating acceleration and slowing of transport.}
  \label{fig:p5}
\end{figure}

\section{Accounting for backward mutations (numerics)}

When $\mu\neq 0$ and $F$ is given by the nonlocal form \eqref{eq:F_nonlocal}, \eqref{eq:advdiff_mu}–\eqref{eq:ibc} define an initial boundary value problem for a quasilinear parabolic equation with a nonlocal term. Analytical solutions are generally unavailable, so we adopt a numerical approach. 

At the right boundary $x=\bar x$ of a sufficiently large truncated interval $[0,\bar x]$, asymptotic analysis indicates stabilization in $x$, which motivates the Neumann boundary condition
\begin{equation}
\left. \frac{\partial u(x,t)}{\partial x} \right|_{x=\bar{x}} = 0.
\label{eq:neumann}
\end{equation}
To verify adequacy and select $\bar x$, solve a sequence of problems with \eqref{eq:neumann} at $\bar x=x_i$ ($x_1<x_2<\dots<x_k$). Let $u_i$ be the corresponding solutions. For a large $T$, define
\begin{equation}
\Delta_i = \int_0^{\bar{x}_i} \int_0^T \big( u_i(x,t) - u_{i-1}(x,t) \big)^2 \, dx\, dt.
\label{eq:Delta}
\end{equation}
If $\Delta_i$ becomes sufficiently small at some $x_i$, the choice $\bar x=x_i$ is deemed adequate.

As an illustration, consider \eqref{eq:advdiff_mu}, \eqref{eq:ibc}, \eqref{eq:F_nonlocal} with
\[
\varphi(x) = \varepsilon e^{-(x/\varepsilon)^2}, \quad 
g(t) = \frac{\varepsilon + \varepsilon_1 t}{1+t}, \quad 
r(x) = \frac{\delta x}{1+x},
\]
parameters $\varepsilon = 0.5$, $\delta=0.15$, $v_0=0.3$, $\gamma=0.3$, and compare level curves for $\varepsilon_1=0.3$ versus $\varepsilon_1=0$ (see Fig.~\ref{fig:p6}). 

Another dynamics, for a non-monotone growth of
\[
r(x) = \frac{\delta x}{1+\delta_1 x^2}, \qquad \delta_1=0.1,
\]
is  is shown in Fig.~\ref{fig:p7}.

\begin{figure}[!h]
  \centering
  \begin{subfigure}[t]{0.45\linewidth}
    \centering
    \includegraphics[width=\linewidth]{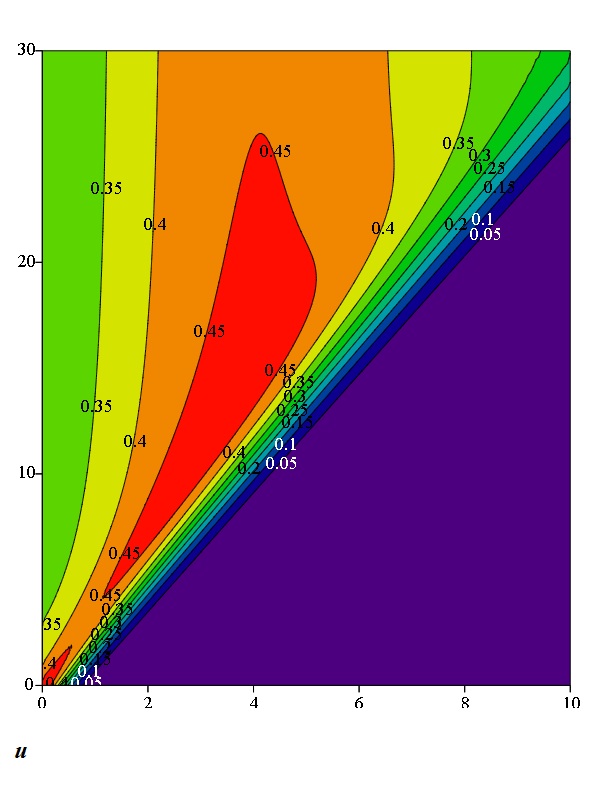}
    \caption{Case:  $\varepsilon_1=0.3$, inflow grows linearly at first, then saturates.}
    \label{fig:p5a}
  \end{subfigure}
  \hfill
  \begin{subfigure}[t]{0.45\linewidth}
    \centering
    \includegraphics[width=\linewidth]{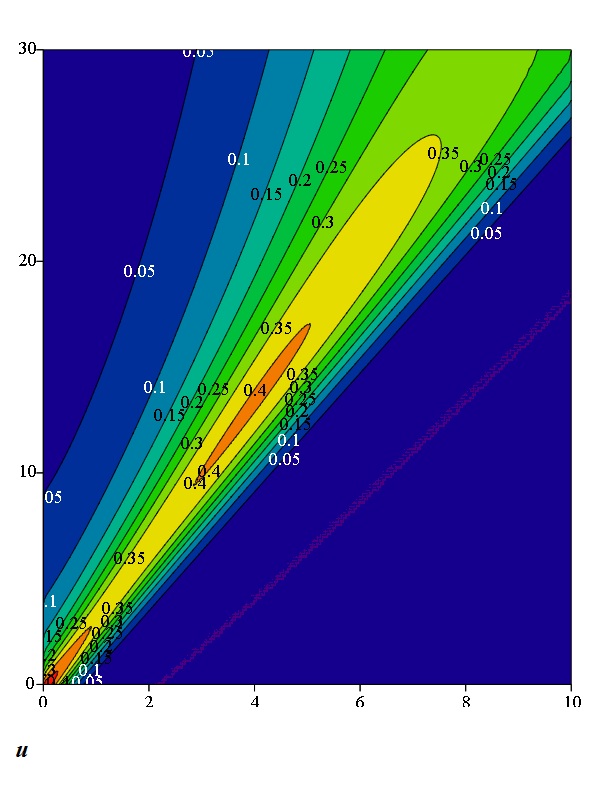}
    \caption{Case: $\varepsilon_1=0$, inflow is  decaying.}
    \label{fig:p5b}
  \end{subfigure}
  \caption{Level sets of the solution $u(x,t)$ under different boundary inflows \(g(t)\).
Both systems share the same growth function 
\(r(x)=\tfrac{\delta x}{1+x}\), transport velocity \(v(t)=v_0\),
and saturation.}
  \label{fig:p6}
\end{figure}

\begin{figure}[!h]
  \centering
  \begin{subfigure}[t]{0.45\linewidth}
    \centering
    \includegraphics[width=\linewidth]{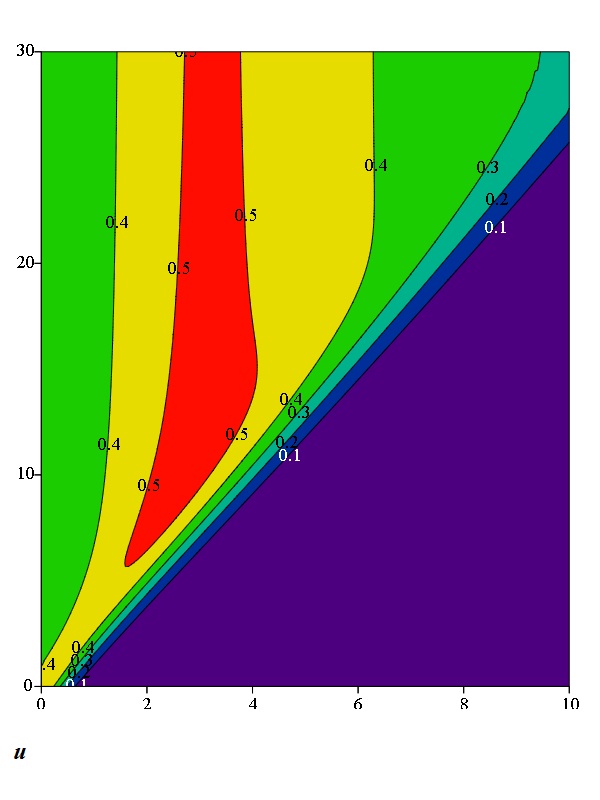}
    \caption{Case:  $\varepsilon_1=0.3$, inflow grows linearly at first, then saturates.}
    \label{fig:p6a}
  \end{subfigure}
  \hfill
  \begin{subfigure}[t]{0.45\linewidth}
    \centering
    \includegraphics[width=\linewidth]{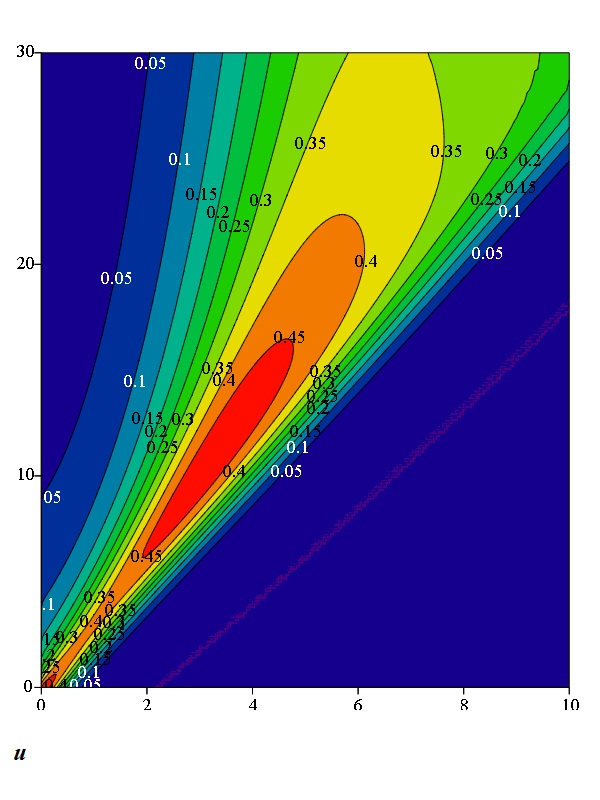}
    \caption{Case: $\varepsilon_1=0$, inflow is  decaying.}
    \label{fig:p6b}
  \end{subfigure}
  \caption{ Comparison of solutions under different boundary inflows \(g(t)\). Both share the same transport \(v(t)=v_0\) and growth law 
\(r(x)=\tfrac{\delta x}{1+\delta_1 x^2}\).}
  \label{fig:p7}
\end{figure}

\section{Conclusion}

The model presented in this article offers a comprehensive framework for mutation–selection dynamics by integrating time-varying mutation rate, logistic growth constraints, nonlocal competition, and reversible mutation (modelled as diffusion in trait space). It generalises the classical quasispecies equations of Eigen and the Crow--Kimura model by allowing both forward and backward mutation flows, while accounting for finite genotype space and population saturation effects. In simplified parameter cases, the model reproduces known results: when mutation is absent, it reduces to logistic growth of the fittest clone, whereas under constant mutation it yields a steady-state mutation--selection balance, analogous to the quasispecies equilibrium described in earlier works~\cite{Eigen1971, Nowak1992, Bocharov2000}.  We assume that incorporating these effects brings the model closer to realistic representations of biological evolution, enabling it to describe and interpret complex phenomena observed in virology, microbiology, and cancer research. By introducing time-varying mutation rates, the model provides a natural framework to explore stress-induced mutagenesis and drug-driven evolution, bridging evolutionary theory and experimental observations in systems where organisms dynamically modulate their mutation rates~\cite{Shibai2017, Loeb2011}.

In cancer modelling, our approach aligns well with existing work on chemotherapy-induced resistance. For example, Cho \textit{et al.} (2018) developed models that explicitly incorporate drug-driven selection of resistant phenotypes in solid tumours, which is a natural application domain for our framework \cite{Cho2018}. The conceptual distinction between spontaneous and treatment-induced mutation resistance advanced by Greene \textit{et al.} also maps directly to our time-dependent mutation formulation \cite{Greene2019}. Moreover, the rich body of literature on metronomic chemotherapy, tumour heterogeneity, and control of evolving resistance provides a concrete bridge from theoretical modelling to clinical design \cite{Ledzewicz2016,Lorz2013}.

In viral evolution, our model can help formalise insights from recent empiricalformalise In \cite{Luo2024}, Luo proposes evolutionary constraints on SARS-CoV-2 mutation trajectories, offering realistic fitness landscapes for viral modelling. Tai \textit{et al.} (2025) study the trade-off between within-host mutation frequency and between-host transmission fitness, which highlights selective pressures that a dynamic mutation model can capture \cite{Tai2025}. More insights on time-varying, flow-based mutations can be driven from the empirical observations of persistent SARS-CoV-2 infections, where mutations accumulate over time \cite{Maarif2025}. Another example is the directed evolution of virus-like particles  \cite{Raguram2024}, which demonstrates how evolution of viral traits under laboratory selection can be guided by mutation–selection models.

In microbial adaptation, the concepts of mutation flux and competition under limited resources are well established \cite{Elena2003, Barrick2009}. Our model could simulate in vitro evolution experiments in bacteria under varying stress or antibiotic pressure, especially where mutagenesis is deliberately toggled over time. Observations of reversion of antibiotic resistance in absence of drug suggest that reversible mutation (diffusion in genotype space) \cite{Andersson2010,  Levin2000} is biologically plausible and important to capture.

Looking forward, the model discussed in this paper holds promise for empirical validation. In microbial evolution experiments, one could impose time-varying mutagenic environments caused by different interventions (e.g. pulses of UV, therapeutic agents or chemical mutagens) and track mutant frequency distributions over time by sequencing. In virology, longitudinal sampling of viral populations during antiviral treatment or chronic infection could test predictions about mutation front dynamics and saturation. In oncology, tracking tumour subclonal composition under adaptive therapy protocols might reveal patterns consistent with our model’s predicted equilibrium or cycling behaviour.

Future extensions could incorporate stochasticity (to capture drift or rare mutational events), spatial structure or tissue architecture (so competition and mutation flows are not globally mixed), and coupling with immune or environmental feedbacks (e.g. how host immunity or therapy responds adaptively to evolving populations). By combining analytical tractability with flexible biological realism, this model serves as a bridge between abstract evolutionary theory and experimentally testable predictions.

\section*{Acknowledgments}
A.~S.~Bratus was supported by the Russian Science Foundation under Grant No.~23-11-00116 
(\url{https://rscf.ru/project/23-11-00116/}, accessed on 15~October~2025) and funded by 
the Moscow Center for Fundamental and Applied Mathematics at Lomonosov Moscow State University 
under Agreement No.~075-15-2025-345.

\section*{Disclosure statement}
The authors declare that this research was conducted without any commercial or financial relationships that could be construed as a potential conflict of interest.
\newpage
\newpage

\end{document}